%% file: main.tex
\documentclass[conference]{IEEEtran}

\IEEEoverridecommandlockouts

\input{header_traditional}

\usepackage{amsmath,amsfonts}
\usepackage{dsfont}
\usepackage[utf8]{inputenc}
\usepackage{amsthm}
\usepackage{svg}
\usepackage{enumitem}

\newcommand{\one}{\mathds{1}}

\newcommand{\ket}[1]{\lvert #1 \rangle}

\newcommand{\DHL}[1]{\mathfrak D(\mathcal H_{#1})}

\theoremstyle{plain}
\newtheorem{thm}{Theorem}[section]
\theoremstyle{plain}

\theoremstyle{plain}
\newtheorem{lem}[theorem]{Lemma}
\theoremstyle{plain}

\theoremstyle{remark}
\newtheorem*{rem}{Remark}
\theoremstyle{definition}

\begin{document}

\title{Quantum Feedback Cooling without State Filtering}

\author{Lorenzo Franceschetti, Francesco Ticozzi
\thanks{The authors are with the Department of Information Engineering, University of Padova, Italy. Emails: \texttt{lorenzo.franceschetti@studenti.unipd.it}, \texttt{ticozzi@dei.unipd.it}.}
\thanks{F.Ticozzi was supported by European Union through
NextGenerationEU, within the National Center for HPC, Big Data and
Quantum Computing under Projects CN00000013, CN 1, and Spoke
10.}}

\maketitle

\begin{abstract}
We introduce a state-based feedback law that stabilizes quantum states or subspaces associated with extremal values of a continuously monitored observable - a problem motivated by quantum cooling tasks. We then propose an output-based approximation that uses simple filtering of the measurement record to emulate the required feedback signal, thereby avoiding full real-time quantum state estimation, a key bottleneck for implementing and scaling filtering-based feedback control. The performance of the  resulting strategy is demonstrated numerically on two test-bed models for feedback cooling.
\end{abstract}


\thispagestyle{empty}
\pagestyle{empty}

\section{Introduction}
Feedback techniques represent a natural choice to obtain robust preparation of desired target states for quantum information processing and quantum cooling. The typical setup \cite{AltafiniIntroduction2012,wisemanmilburn,barchielli} is obtained by indirectly and continuously monitoring the target system and implementing a control action dependent on the measurement outcome. The continuous monitoring, due to measurement backaction, effectively makes conditional state evolution driven by a stochastic master equation \cite{wisemanmilburn,bouten,pellegrini,barchielli}. Two main design approaches have been developed: Wiseman and Milburn \cite{wiseman-markovfb} proposed a direct feedback scheme that can be interpreted as a {\em derivative feedback} controller \cite{gough-pid}, whose average evolution is a time-invariant semigroup and whose effectiveness has been characterized from a control theoretic-viewpoint in \cite{ticozziAnalysisSynthesisAttractive2009}. In particular, it is shown that it cannot prepare eigenstates of the measurement operator - a natural target in various contexts, for example when monitoring the system energy (its Hamiltonian) to obtain feedback cooling.
In contrast, {\em filtering-based feedback} (see e.g. \cite{bouten,AltafiniIntroduction2012} for a review of the theoretical framework) has been pioneered by Belavkin \cite{belavkin} and has been later developed into to a variety of control approaches \cite{mirrahimi, globstabfeed,destab,expstabqnd,weichaoexp,switching}. In this framework, the control law is able to stabilize {\em only} the measurement eigenstates. While filtering-based quantum feedback control has been successfully employed experimentally \cite{haroche}, its potential as a scalable control tool is seriously hindered by the necessity of obtaining a state estimate in real-time. Due to the quadratic scaling of the state in the system degrees of freedom, and its exponential scaling on the number of subsystems, it quickly becomes an unattainable task. An alternative approach has been proposed in \cite{expstabqnd}, where the control law does not require a full-state estimate, but only an estimate of the population of the eigenstates of the measurement operators. While reducing the computational burden, it still requires the filtering of a number of variables that is exponential in the number of subsystems, so it is not feasible for real-time control of larger systems.

In this work, we focus on the {\em measurement-eigenspace} stabilization problem and show that full-state reconstruction is not necessary when the target eigenspace is extremal, that is, when the associated eigenvalue is either the minimum or the maximum eigenvalue of the measurement operator. Optimal cooling and heating problems fit directly into this setting.
The contributions are two fold: we first develop a filtering-based control strategy that makes the target globally asymptotically stable, based only on {\em the expectation of the monitored observable}; and we next show how the latter can be approximated in practice by {\em a simple moving average} of the output signal. The approximation is asymptotically exact for the open loop dynamics, while its effectiveness under the controlled dynamics motivated by an ergodicity argument and  demonstrated via simulations. The control approach adopts a switching structure inspired by \cite{mirrahimi,destab}, adapted to the available information, and extended to ensure {\em subspace} stabilization. The methods are tested numerically on two prototypical examples.

\section{Model and Assumptions}

Let $\mathcal H$ be a Hilbert space and call $\mathfrak B(\mathcal H)$ the set of bounded linear operators acting on $\mathcal H$ and $\mathfrak h(\mathcal H)$ the subset of $\mathfrak B(\mathcal H)$ of hermitian operators. The state of the quantum system associated with $\mathcal H$ is represented by a \emph{density matrix} $\rho \in \mathfrak D(\mathcal H)$, where $\mathfrak D(\mathcal H) = \{\rho \in \mathfrak h(\mathcal H) | \rho \ge 0, Tr[\rho] = 1\}$ is the set of density matrices on $\mathcal H$. Given a decomposition of $\mathcal H = \mathcal H_S \oplus \mathcal H_R$, we define the set of states having support only on $\mathcal H_S$ as $\mathfrak D(\mathcal H_S) = \{\rho \in \mathfrak D(\mathcal H)\ |\ \Pi_S \rho = \rho\}$, where $\Pi_S$ is the orthogonal projection matrix onto $\mathcal H_S$. 

This work focuses on an $N$-level time-invariant open quantum system, associated with an $N$-dimensional Hilbert space $\mathcal H$, subject to unit-efficiency, continuous homodyne-type measurements, leading to a diffusive stochastic master equation (SME) \cite{AltafiniIntroduction2012,barchielli,wisemanmilburn}: the time evolution of the density matrix $\rho_t \in \mathfrak D(\mathcal H)$ can be expressed as a stochastic differential equation driven by a Wiener process of the form
\begin{equation}
    d\rho_t = -i[H_0 + f_tF_0, \rho_t]dt + \mathcal D_L[\rho_t]dt + \mathcal G_L[\rho_t]dW_t \label{model:ctrl}
\end{equation}
where $H_0 \in \mathfrak{h}(\mathcal H)$ is the free Hamiltonian of the system, $F_0$ a control Hamiltonian, $f_t$ a scalar control function, and $L \in \mathfrak B(\mathcal H)$ represents the coupling between the system and the measurement apparatus. The stochastic process $W_t$ is a Wiener process on a fixed filtered probability space $(\Omega, \mathcal F, \{\mathcal F_t^W\}, \mathbb P)$ and describes the quantum backaction due to continuous measurement on the system.
The measurement apparatus produces a measurement signal $y_t$, whose dynamics is described as \(
    dy_t = Tr[(L + L^\dagger)\rho_t]dt + dW_t.
\)

In the rest of the paper, we assume that \begin{enumerate}
    \item [$(A_1)$] $L$ is a \emph{hermitian} operator;
    \item [$(A_2)$] $[H_0, L] = 0.$
\end{enumerate} 
These assumptions generalize the paradigmatic case $H_0=L,$ where the energy is the quantity being measured, adding flexibility to the design. We call $\{\lambda_i\}_1^r$ the $r \le N$ distinct eigenvalues of $L$ in increasing order and $\mathcal H_\lambda$ the eigenspace of $L$ relative to its eigenvalue $\lambda$.

Our aim is to design an efficient control law that is able to globally asymptotically stabilise either $\DHL{\lambda_1}$ (\emph{cooling}) or $\DHL{\lambda_r}$ (\emph{heating}). Depending on the task, we add one of the two following spectral-separation assumptions: \begin{enumerate}
    \item [$(A_3^c)$] $Tr[L] < N\lambda_2$ in the cooling scenario
    \item [$(A_3^h)$] $Tr[L] > N\lambda_{r - 1}$ in the heating scenario
\end{enumerate}
{It is easy to construct an $L$ satisfying the assumptions above by shifting the non-extreme eigenvalues: whether or not it is a physically accessible observable is to be determined depending on the specific application scenario.}

\section{A Switching Control law for Cooling}
\label{idcontr}
The control law we discuss in the following exploits the open-loop behaviour of the system, monitoring its evolution and destabilising it whenever we detect the state is not converging to the desired set of states. Hence, we first need to assess the asymptotic behaviour of equation (\ref{model:ctrl}) when $f_t = 0$. {Each set} $\DHL{\lambda}$ is invariant with probability $1$ under assumptions $(A_1)$ and $(A_2)$. {This can be shown using Theorem (2.1) in \cite{Benoist2017}, noting that, in the block partition induced by  $\DHL{\lambda}$, the operators $H_0$ and $L$ are both block diagonal (thanks to $(A_2),$ they can actually be diagonalized simultaneously).} 
It can further be shown that $\rho_t$ converges with probability $1$ to one of the sets $\DHL{\lambda}$\footnotemark{} and that, given the initial condition $\rho_0$, the probability that $\rho_t$ converges to $\DHL{\lambda}$ is $Tr[\Pi_\lambda \rho_0]$, where $\Pi_\lambda$ is the orthogonal projection matrix onto $\mathcal H_\lambda$.
\footnotetext{The proof is similar to the one found in \cite{convergence}, in our case applied to the dynamical equation (\ref{model:ctrl}) with $f_t$ set to $0$. In particular, the function $V_t$ defined in equation (23) in \cite{convergence} becomes $V_t = Tr[L^2\rho_t] - Tr[L\rho_t]^2$, which is the variance process of $L$.}

Assume for now that we have access to the signal \(
    x_t = 2Tr[L\rho_t]
\),
the drift component of the actual output signal $y_t$. This signal can be used to distinguish whether $\rho_t$ is converging to one of the two extremal sets, $\DHL{\lambda_1}$ or $\DHL{\lambda_r}$, at least asymptotically. Indeed, by noting that $x_t/2$ is a convex combination of the eigenvalues of $L$, $\rho_t$ converges to $\DHL{\lambda_1}$ (resp. $\DHL{\lambda_r}$) {\em if and only if} $x_t$ converges to $2\lambda_1$ (resp. $2\lambda_r$).

Since $\rho_t$ converges with probability $1$ to one of the $\DHL{\lambda}$, $x_t$ is enough to distinguish whether we are converging to the desired set or not. We then consider a controlled system whose evolution is given by \eqref{model:ctrl}
where $f_t=f(x_t)$ is a static deterministic function of $x_t.$ 

We focus on the ``cooling'' task to simplify the presentation: conversion to ``heating'' is straightforward. In this setting, the control law is provided in the following theorem: \begin{thm}
    Under assumptions $(A_1), (A_2)$ and $(A_3^c)$, consider the controlled system (\ref{model:ctrl}), with $H_0, L, F_0 \in \mathfrak h(\mathcal H)$. Let $V(\rho) = Tr[L\rho] - \lambda_1$ and $\delta = \lambda_2 - \lambda_1$. Let $\gamma \in (0, \delta)$ and define $f(x)$ as follows: \begin{itemize}
        \item $f(x) = 1$ if $x \ge 2\delta - \gamma + 2\lambda_1$
        \item $f(x) = 0$ if $x \le 2\delta - 2\gamma + 2\lambda_1$
        \item if $x \in (2\delta - 2\gamma + 2\lambda_1, 2\delta - \gamma + 2\lambda_1)$, then $f(x) = 1$ if $x$ last entered the interval from above, else $f(x) = 0$.
    \end{itemize}
    Assume that $F_0$ is such that, when $f(x) = 1$, the associated average equation \begin{equation}
        \dot \sigma_t = \mathscr L[\sigma_t] = -i[H_0 + F_0, \sigma_t] + \mathcal D_L[\sigma_t] \label{lind:avg}
    \end{equation}
    has $\rho_{mix} = \frac{1}{N}\one$ as its unique equilibrium. Then there exists a $\gamma \in (0, \delta)$ such that $\rho_t$ almost surely (a.s.) converges to 
    $\mathfrak D(\mathcal H_{\lambda_1}^L).$
    \label{convthm}
\end{thm}
Before providing the proof, some additional remarks and auxiliary results are in order: (i) the third switching condition prevents chattering; 
(ii) if $L$ is hermitian, $\rho_{mix}$ is always an equilibrium for \eqref{lind:avg}, independently of  $F_0$. Hence, we need to design it to ensure that no other equilibrium is present. When this is possible within the active control constraints, the following lemma ensures that picking a random $F_0$ is sufficient. 
\begin{lem}
    Under assumptions $(A_1)$ and $(A_2)$, consider a parametric $F_0(w) : \mathbb R^n \rightarrow \mathfrak h(\mathcal H)$, whose components are analytic in $w \in \mathbb R^n$. If there exists $\bar w \in \mathbb R^n$ such that $\rho_{mix}$ is the unique equilibrium for (\ref{lind:avg}), then it is the unique equilibrium for almost all $w \in \mathbb R^n$. 
\end{lem} \begin{proof}
    Consider $\mathscr L_w[\sigma_t] = -i[H_0 + F_0(w), \sigma_t] + \mathcal D_L[\sigma_t].$
    
    The state $\rho_{mix}$ is always an equilibrium of the average dynamics, so (in matrix representation) $\mathrm {rank}\mathscr L_w \le N^2 - 1$ for all $w$; uniqueness of $\rho_{mix}$ as equilibrium is equivalent to $\mathrm {rank } \mathscr L_w = N^2 - 1$.
    By assumption, there exists $\bar w \in \mathbb R^n$ such that $\mathrm{rank } \mathscr L_{\bar w} = N^2 - 1$. By Lemma (5) in \cite{zeromeas}, the set of $w \in \mathbb R^n$ for which the rank is strictly less than $N^2 - 1$ has zero measure with respect to the Lebesgue measure.
\end{proof}

The additional supporting lemmas needed in the proof are reported in Appendix (\ref{app}).

\begin{proof}[Proof of Theorem (\ref{convthm})]
    {Define the set \begin{equation}
        \mathfrak D_{> \alpha} = \{\rho \in \mathfrak{D}(\mathcal H) | V(\rho) > \alpha\}.
    \end{equation} Since condition $\rho_t \in \mathfrak D_{> \alpha}$ is equivalent to $Tr[L\rho_t] - \lambda_1 > \alpha$,  we obtain
\(
        \rho_t \in \mathfrak D_{> \alpha} \iff x_t > 2\alpha + 2\lambda_1.
 \)
     We can thus rewrite the control function as \begin{itemize}
        \item $f(\rho) = 1$ if $\rho \in \mathfrak D_{\ge \delta - \frac{\gamma}{2}}$
        \item $f(\rho) = 0$ if $\rho \in \mathfrak D_{\le \delta - \gamma}$
        \item if $\rho \in \{\rho \in \mathfrak D(\mathcal H) | V(\rho) \in (\delta - \gamma, \delta - \frac{\gamma}{2})\}$, then $f(\rho) = 1$ if $x$ last entered the set from $\mathfrak D_{\ge \delta - \frac{\gamma}{2}}$, else $f(\rho) = 0$.
    \end{itemize}}

    By Lemma (\ref{lem45mod}), there exists a $\gamma \in (0, \delta)$ such that, when $f(\rho_t) = 1$ and $\rho_t \in \mathfrak D_{> \delta - \gamma}$, the trajectory exits the superlevel set in finite time with probability $1$.

    Let $\mathbb T$ be the ordered sequence of time instants where $f(\rho_t)$ switches from $1$ to $0$. Each instant $t \in \mathbb T$ represents the first instant when $\rho_t$ enters $\mathfrak D_{\le \delta - \gamma}$ having previously entered $\mathfrak D_{\ge \delta - \frac{\gamma}{2}}$. The state trajectory takes a finite time to go from $\mathfrak D_{\le \delta - \gamma}$ to $\mathfrak D_{\ge \delta - \frac{\gamma}{2}}$, since $\rho_t$ is a continuous function and it has to traverse a finite non-zero measure set, so we can conclude that $\mathbb T = \{t_k\}$ is a countable set.

    Assume by contradiction that $\rho_t$ never converges to $\DHL{\lambda_1}$. With probability $1$, this happens if and only if at each instant $t_k$ the trajectory $\rho_t$ enters $\mathfrak D_{\ge \delta - \frac{\gamma}{2}}$ in finite time. Indeed, all equilibria for the uncontrolled dynamics that are not in $\DHL{\lambda_1}$ live in $\mathfrak D_{\ge \delta}$ and by Lemma (\ref{lemfiniteentrance3}) the trajectory must enter $\mathfrak D_{\ge \delta - \frac{\gamma}{2}}$ in finite time with probability $1$. This in turn is equivalent to $\mathbb T$ being infinite, indeed if there were a finite number of instants in $\mathbb T$, then the trajectory would converge to $\DHL{\lambda_1}$ with probability $1$ as desired.

    Let $A_k$ be the event representing the fact that $\rho_t$ enters $\mathfrak D_{\ge \delta - \frac{\gamma}{2}}$ in finite time starting from $\rho_{t_k} \in \mathfrak D_{\le \delta - \gamma}$. The probability of never converging to $\DHL{\lambda_1}$ is equal to the probability of the intersection of all $\{A_k\}_{k \ge 1}$, which can be bounded using the law of total probability as follows: \begin{align}
        \mathbb P \left[\bigcap_{k = 1}^{+\infty} A_k\right] \nonumber  &= \prod_{k = 1}^{+\infty} \mathbb P[A_k | A_1, \dots, A_{k - 1}] \nonumber \\ &
        \le \prod^{+\infty}_{k = 1} \left(1 - \frac{\gamma}{2\delta - \gamma}\right) = 0
    \end{align}
    where we have used Lemma (\ref{upboundprob}) to bound the conditional probability $\mathbb P[A_k | A_1, \dots A_{k - 1}]$, noting that conditioning on $\{A_1, \dots, A_{k - 1}\}$ is equivalent to assuming that $\rho_{t_k} \in \mathfrak D_{\le \delta - \gamma}$.    

    We can conclude that $\rho_t$ converges to $\DHL{\lambda_1}$ with probability $1$.
    \end{proof}
\begin{rem}
The proposed control law also makes $\DHL{\lambda_1}$ stable. This can be shown using the same reasoning as in the first point of the proof of Theorem (4.2) of \cite{mirrahimi} using $u_1(\rho) = 0$. We can conclude that the control law that we have designed is able to render $\DHL{\lambda_1}$ globally asymptotically stable (GAS) \cite{vanhandel}. To convert these results to the ``heating'' case, we need to redefine $V(\rho) = Tr[L\rho] - \lambda_r$, the inequality signs in the definition of $f$ need to be inverted, $\delta$ becomes $\lambda_{r - 1} - \lambda_{r} < 0$ and $\gamma$ is now in the interval $(\delta, 0)$.
\end{rem}

\section{Implementation without state filtering}
Our control law is able to render the ground subspace GAS, provided that the signal $x_t$ is available. However, in actual applications, this signal is not directly available: one (computationally heavy) way to obtain it would be to estimate the {\em full} state $\rho_t$ using the output signal $y_t$, that is, to use filtering-based feedback \cite{wisemanmilburn, AltafiniIntroduction2012}. 

A different strategy consists in estimating $x_t$ directly from $y_t$. Consider the {\em open-loop evolution}: the state asymptotically converges \cite{convergence} to one of the invariant sets $\DHL{\lambda}$, and thus the measured signal asymptotically becomes $y_t = 2\lambda t + W_t$, since $Tr[L\rho] = \lambda$ for all $\rho \in \DHL{\lambda}.$ We can then obtain an (``ergodic'') estimate of $x_t$ as \begin{equation}
    \hat x_t = \frac{y_t}{t} = \bar x_t + \tilde W_t
    \vspace{-3mm}
\end{equation}
where we define $\bar x_t=t^{-1}\int_0^tx_\tau d\tau$ and $\tilde W_t = \frac{W_t}{t}$. 
The Gaussian random variable $\tilde W_t$ has zero mean and variance $\frac{1}{t}$. This means that $\hat x_t$ converges asymptotically to $x_t=2\lambda$  in the long time limit. Obtaining such a signal requires minimal processing, especially compared to integrating the SME in real time.
The first approach we consider thus consists in replacing the $x_t$ signal in our control law with the signal $\hat x_t$. To ensure that the controller does not operate when the noise is dominant (for $t$ close to $0$, the variance of $\tilde W_t$ tends to infinity), we introduce an initial phase where the control function $f$ is set to $0$ independently of what we measure. In particular, we wait to activate the controller until $\tau_s$, defined as \begin{equation}
    \tau_s = \inf_{t > 0} \{t\ |\ \mathbb P[|\tilde W_t| \ge \varepsilon] \le 1 - \beta\} \label{taus}
\end{equation}
where $\beta \in (0, 1)$ and $\varepsilon > 0$ are two design parameters that influence the initial delay and require tuning. In particular, we can compute it in closed form as $\tau_s = \left[\frac{1}{\varepsilon}\Phi^{-1}\left(\frac{1 + \beta}{2}\right)\right]^2$, where $\Phi$ is the cumulative distribution function of a standard Gaussian random variable.

While being a computationally efficient method, it can lead to poor performances. Indeed, we are averaging all the output increments, rendering the estimate less responsive to the dynamics of the system, in particular in the presence of fast transients or rare events. To overcome this issue, we propose a rolling-average scheme, where instead of considering the signal $y_t$, we consider its windowed version: \begin{equation}\label{eq:window}
    y^{\Delta}_t = \int_{\max\{t - \Delta, 0\}}^t dy_t
\end{equation}
where $\Delta$ is the window length. From this, we consider the estimate \begin{equation}
    \hat x^\Delta_t = \frac{y^\Delta_t}{\min\{t, \Delta\}}.
\end{equation}
The $\max$ and $\min$ are used to deal with the initial phase when the window is not yet full.

The windowing allows us to focus on the most recent measurements, which are the most relevant for our purposes. The window length influences the controller behaviour and requires tuning: a shorter window leads to a more responsive controller, at the expense of a larger noise component. In contrast, a longer window reduces the impact of noise, with the downside of reducing responsiveness and increasing the memory footprint, since we need to store all the measured increments in the window to compute the signal $\hat x^\Delta_t$. 

Using these approximations within our feedback strategy is justified, since: (1) in between two control activation periods the system effectively evolves in open loop, so the trajectory must approach one of the equilibria. Hence $\hat x_t, x^\Delta_t$ tend to $x_t=2\lambda$ as desired; (2) When the control is active, by Theorem 4 in \cite{ergmean}, $\frac{1}{t}y_t$ converges a.s. to the value of $x_t$ for the unique equilibrium $\rho_{mix}$, thus entering a.s. the no-control zone due to assumption $(A_3^c),$ as in Lemma \ref{Texists}. A detailed convergence proof goes beyond the scope of this work and will be presented elsewhere.

\section{Simulations}
To evaluate the practical implementation of the control law, we analyze the controlled behavior of two different systems. To satisfy the physical constraints on the density matrix $\rho_t$, we employ the integration scheme described in \cite{rouchon}, with a step size of $\Delta t = 0.0001$.

In the figures, \textit{free evolution} refers to the behavior of the open-loop system. \textit{Ideal signal} represents the behaviour using $x_t$ as described in Section (\ref{idcontr}) and is used as a benchmark reference. \textit{Ergodic signal} refers to the controller that uses the $\hat x_t$ signal, while \textit{windowed signal} refers to the use of the $\hat x_t^\Delta$ signal, with $k$ being the window length in number of samples, so $\Delta = k\Delta t$. The plotted trajectories are obtained as the average of $1000$ trajectories starting from different initial conditions, each of which is the average of $20$ single-run trajectories starting from the same initial condition.

{\em Qutrit system:} The first system that we study is a $3$-dimensional system with Hamiltonian $H_0 = \mathrm{diag}(-1, 2, 3)$ and $L = H_0$. The measurement operator satisfies assumption $(A_3^c)$. The target state, in this basis, is $\rho_d = \mathrm{diag}(1, 0, 0)$. For the feedback operator, we choose 
  \[\vspace{-1mm}  F_0 = \begin{bmatrix}
        0 & 1 & 1 \\ 1 & 0 & 1 \\ 1 & 1 & 0
    \end{bmatrix},\]
which ensures that the only equilibrium for the average dynamics is $\rho_{mix}$. As for the $\beta$ and $\varepsilon$ parameters in the definition (\ref{taus}) of $\tau_s$, we set $\beta = 0.6$ and $\varepsilon = \delta = 3$.

Figure (\ref{3d:fidelity}) reports the average $F(\rho_t) = Tr[\Pi_{\lambda_1}\rho_t]$ trajectory for the different control strategies. $F(\rho) = 1$ if and only if $\rho \in \mathfrak D (\mathcal H_{\lambda_1})$ and $0$ if and only if $\rho$ has no support on $\mathcal H_{\lambda_1}$. 

\begin{figure}
    \centering
    \includegraphics[width=1.0\linewidth]{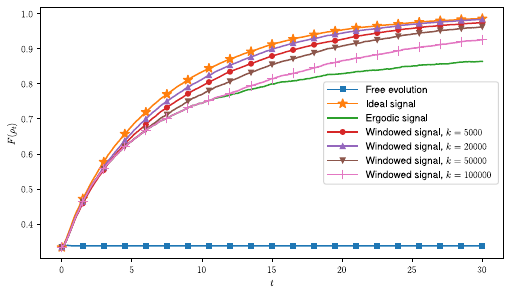}
    \caption{Average $F(\rho_t)$ trajectories for the qutrit system using different control strategies.}
    \label{3d:fidelity}
\end{figure}

All controllers are capable of steering $\rho_t$ towards the right subspace. While the performance of the controller using $\hat x_t$ is lacking, windowed controllers are able to closely approach the ideal behaviour. The figure highlights a non-monotonic dependence on $\Delta$, with an increase in $k$ initially corresponding to improved performance, and then leading to worse behaviour (as expected, since when $\Delta$ increases, $\hat x_t^\Delta$ converges to $\hat x_t$).

\emph{Antiferromagnetic Heisenberg triangle:} The other system of interest is the cooling of an antiferromagnetic Heisenberg triangle, with nearest-neighbour interaction and periodic boundary conditions in the absence of any external magnetic field and subject to continuous homodyne measurement described by the measurement operator $L = H_0$, where the Hamiltonian is computed as \begin{equation}
    H_0 = \sum_{k \in \{x, y, z\}}\left(J_k\sum_{i = 1}^3 S_i^kS_{i + 1\ \mathrm{mod}\ 3}^k\right)
\end{equation}
where we define \(
    S_i^k = \one_2^{\otimes{i - 1}} \otimes \sigma_k \otimes \one_2^{\otimes{3 - i}}
\)
as the operator that acts as the Pauli matrix $\sigma_k$ on site $i$ and as the identity on the rest.

We set $J_x = J_y = 1$ and $J_z = 2$, for which the measurement operator $L$ satisfies assumption $(A_3^c)$. We choose $F_0$ such that, in a basis where $H_0$ is diagonal, it is a tridiagonal matrix with main diagonal set to $0$ and both super- and subdiagonal set to $4$. We choose $\beta = 0.6$ and $\varepsilon = \lambda_2 - \lambda_1 = 6$.
\begin{figure}
    \centering
    \includegraphics[width=1.0\linewidth]{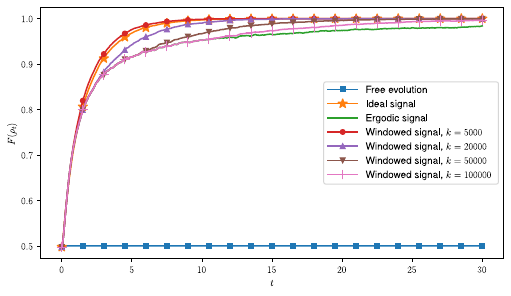}
    \caption{Average $F(\rho_t)$ trajectories for the antiferromagnetic Heisenberg triangle using different control strategies.}
    \label{heis:fidelity}
\end{figure}
Figure (\ref{heis:fidelity}) confirms that our control laws ensure convergence to the desired subspace. We again observe different behaviors for different values of $\Delta$, with $k = 5000$ providing the best performance.

{\em Remark:} Note that $\DHL{\lambda_1}$ is generated by the orthogonal (partially) entangled states $\ket{\psi_1} = \frac{1}{\sqrt{2}}[\ket{010} - \ket{001}]$, $\ket{\psi_2} = \frac{1}{\sqrt{2}}[\ket{101} - \ket{011}]$, $\ket{\psi_3} = \frac{1}{\sqrt{6}}[2\ket{100} - \ket{010} - \ket{001}]$, $\ket{\psi_4} = \frac{1}{\sqrt{6}}[2\ket{110} - \ket{101} - \ket{011}]$. Hence, our simple cooling strategy is also effectively creating non-classical correlations in the system. 

\section{Outlook} In this letter, we propose a quantum feedback scheme that asymptotically prepares an extreme eigenspace of the monitored observables, and suggest a minimal output filter to replace the state one.
The proposed methods approximate an ideal reference situation, where the instantaneous measurement average is available and for which stochastic convergence is proved rigorously, and perform well in simulation - especially after tuning of the approximation parameters (averaging window). We believe that the presented results represent a valuable first step towards scalable feedback design and its application in realistic quantum systems: even the simple 3 spin system of the Heisenberg triangle would be challenging to implement with real-time full state filtering.  In addition, while we presented results using only unit-efficiency measurements, simulations confirm its validity also when dealing with non-ideal measurements, despite a slightly slower convergence speed.
Different extensions of the method are being considered, including non-Hermitian measurement operators, integration with derivative control, analysis of the convergence rate for the approximate methods, generalization to the stabilization of non-extremal set of states, and application to counting-type measurements.

\appendix
\section{Supporting lemmas}
\label{app}
The lemmas in this section are all proved under the assumptions $(A_1), (A_2), (A_3^c)$ and the hypotheses of Theorem (\ref{convthm}).

\begin{lem}
    Set $f_t= 0$ and let $\alpha \in [0, \delta)$. Then, if $\rho_t$ enters $\mathfrak D_{\ge \alpha}$, it must do so in finite time with probability $1$.
    \label{lemfiniteentrance3}
\end{lem}
\begin{proof}
    Recall that the free evolution will converge to one of the $\DHL{\lambda}$ sets with probability $1$, and thus $V(\rho)$ has a limit almost surely. Assume first that the trajectory $\rho_t$ converges to one of the other sets $\mathfrak D(\mathcal H^L_{\lambda})$. 
    Define now the set of elementary events \begin{equation}
        \Omega' = \left\{\omega \in \Omega \,\middle| \lim_{t \rightarrow +\infty} \rho_t \not \in \DHL{\lambda_1}\right\}
    \end{equation}
    and the stopping time \(
        \tau(\omega) = \inf\left\{t \middle| \rho_t(\omega) \in \mathfrak D_{\ge \alpha}\right\}.
    \)
    The statement of the lemma is then equivalent to \begin{equation}
        \mathbb P \left[\{\omega \in \Omega' \middle| \tau(\omega) = +\infty\}\right] = 0 \label{zeromeas}
    \end{equation}

    For almost all $\omega \in \Omega'$, we have that \begin{equation}
        \lim_{t \rightarrow +\infty} V(\rho_t(\omega)) = V\left(\rho_\infty(\omega)\right) = v(\omega)
    \end{equation}
    with $v(\omega) = \lambda - \lambda_1 \ge \delta$ almost surely, where $\lambda \neq \lambda_1$ is one of the eigenvalues of $L$. For all $\omega \in \Omega'$ such that this is true, we have that $\forall \varepsilon > 0, \exists N > 0$ such that, for $ t > N$, $|V(\rho_t(\omega)) - v(\omega)| < \varepsilon$, which can be rewritten as \begin{equation}
        v(\omega) - \varepsilon < V(\rho_t(\omega)) < v(\omega) + \varepsilon
    \end{equation}
    Let $\varepsilon = \delta - \alpha > 0$. Since $v(\omega) \ge \delta$, there exists $\bar N > 0$ such that $t > \bar N$ implies $V(\rho_t(\omega)) > \alpha$ for almost all $\omega \in \Omega'$.

    For all $\omega \in \Omega'$ for which $v(\omega) \ge \delta$, there exists a finite time instant $\bar t$ for which $\rho_{\bar t}(\omega) \in \mathfrak D_{\ge \alpha}$. By definition of $\tau(\omega)$, we have that $\tau(\omega) \le \bar t < +\infty$. But the set of $\omega \in \Omega'$ such that this happens has measure $1$ under the probability measure $\mathbb P[\cdot | \rho_\infty \not \in \DHL{\lambda_1}]$, so equation (\ref{zeromeas}) is verified.

    If we instead assume that $\rho_t$ converges to $\DHL{\lambda_1}$, with a similar reasoning, we can again prove that the trajectory must enter $\mathfrak D_{\ge \alpha}$ in finite time.
\end{proof}

We can also provide a bound on the probability of $\rho_t$ entering $\mathfrak D_{\ge \delta - \frac{\gamma}{2}}$ knowing the initial condition $\rho_0\in\mathfrak D_{\le \delta - \gamma}$: \begin{lem}
    If $\rho_0 \in \mathfrak D_{\le \delta - \gamma}$, then \begin{equation}
        \mathbb P\left[\sup_{0 \le t < +\infty} V(\rho_t) \ge \delta - \frac{\gamma}{2}\right] \le 1 - \frac{\gamma}{2\delta - \gamma} < 1
    \end{equation}
    \label{upboundprob}
\end{lem}
\begin{proof}
    The proof is an application of Theorem (2.2) in \cite{mirrahimi}, with $\lambda = \alpha = \delta - \frac{\gamma}{2}$. Since $\rho_0 \in \mathfrak D_{\le \delta - \gamma}$, the control function $f$ is zero until $\rho_t$ enters $\mathfrak D_{\le \delta - \frac{\gamma}{2}}$, so the derivative of $V(\rho)$ along the trajectories is $\mathscr L V(\rho) = 0$. Hence, we can apply Theorem (2.2), bounding $V(\rho_0)$ by $\delta - \gamma$.
\end{proof}

When the control function $f(x)$ is set to $1$, the state $\rho_t$ exits from the set $\mathfrak D_{\ge \delta - \gamma}$ in finite time with probability $1$ .
\begin{lem}
    If $f_t = 1$ for all times, there exists a finite time $T > 0$ such that \begin{equation}
        \zeta_T(\rho) = \min_{t \in [0, T]}\mathbb E[V(\rho_t)] < \delta, \forall \rho \in \mathfrak D_{\ge \delta} \label{zetafn}
    \end{equation}
    with $\rho_t$ starting from $\rho_0 = \rho$.
    \label{Texists}
\end{lem}

\begin{proof}
    Being $\mathbb E[V(\rho)] = V(\mathbb E[\rho])$, we can rewrite the function $\zeta_T(\rho)$ as \begin{equation}
        \zeta_T(\rho) = \min_{t \in [0, T]}V(\sigma_t)
    \end{equation}
    where $\sigma_t = \mathbb E[\rho_t]$ is the average trajectory of the system starting from $\sigma_0 = \rho_0 = \rho$, evolving according to (\ref{lind:avg}).

    $V(\sigma_t)$ is a continuous function of time, being both $V(\rho)$ and $\sigma_t$ continuous functions. By Weierstrass theorem, it admits a minimum on any compact subset of $\mathbb R$. Any interval $[0, T]$ is a compact subset of $\mathbb R$ if $T$ is finite, so $\zeta_T(\rho)$ is well-defined.

    Under the assumption that $\rho_{mix}$ is the unique equilibrium of the Lindblad equation, it can be shown (see e.g. Theorem (1) of \cite{uniqueness}) that it is also globally asymptotically stable. Combining this with the continuity of $V$, we obtain that \begin{equation}
        \lim_{t \rightarrow +\infty} V(\sigma_t) = V(\rho_{mix}) = \delta_1 < \delta
    \end{equation}
    where $\delta_1 < \delta$ by assumption $(A^c_3)$.

    We now consider the function $V_1(\sigma_t) = V(\sigma_t) - \delta_1$, which converges to $0$. It can be rewritten as a linear function of the displacement $\Delta \sigma_t = \sigma_t - \rho_{mix}$: \begin{align}
        V_1(\sigma_t) &= Tr[L(\sigma_t - \rho_{mix})] \nonumber \\ &= Tr[L\Delta \sigma_t] = V_2(\Delta \sigma_t) \label{linearv}
    \end{align}

    The time evolution of $\Delta \sigma_t$ is $\frac{d}{dt}\Delta \sigma_t = \mathscr L[\Delta\sigma_t]$. Using the linearity of $V_2(\Delta \sigma)$, we can write its evolution as a linear combination of the system modes, so we can bound the absolute value of $V_2(\Delta \sigma)$ as 
    \(
        |V_2(\Delta \sigma_t)| \le \sum_i c_t^i e^{\mu_i t} \)
    where $\{\mu_i\}$ are the eigenvalues of $\mathscr L$ and $c^i_t$ are polynomial functions of time, continuously depending on the initial condition $\rho_0$. In particular, $\mu_1 = 0$ and $c^1_t = 0$, while all the other eigenvalues have strictly negative real part, since the system is asymptotically stable. The set $\mathfrak D_{\ge \delta}$ is closed, since it is a super-level set of a positive semidefinite function, and bounded, since it is a subset of $\mathfrak D(\mathcal H)$, itself bounded, hence it is a compact set. By Weierstrass theorem, the polynomia $c_t^i$ can be maximised with respect to $\rho_0 \in \mathfrak D_{\ge \delta}$, leading to \begin{equation}
        |V_2(\Delta \sigma_t)| \le g(t) = \sum_i c_t^{i, max}e^{\mu_it}. \label{boundV2}
    \end{equation}
    
    Let $\delta_2 = \delta - \delta_1 > 0$. If we impose $|V_2(\Delta \sigma_t)| < \delta_2$, this implies that $V(\sigma_t) < \delta_1 + \delta_2 < \delta$, so it also implies that $\sigma_t$ is outside the set $\mathfrak D_{\ge \delta}$.

    For all $c^{i, max}_t \neq 0$, the corresponding $\mu_i$ has strictly negative real part, so  $c^{i, max}_te^{\mu_it}$ converges asymptotically to $0$, so $g(t)$ converges asymptotically to $0$ and we can find a $\bar T > 0$ such that $g(t) < \delta_2$ for all $t > \bar T$. By expression (\ref{boundV2}), $\sigma_t$ does not belong to $\mathfrak D_{\ge \delta}$ for all $t > \bar T$. This $\bar T$ does not depend on the initial condition $\rho_0$, since we are dealing with the maximised bounding function. So, for all $T > \bar T$, the function $\zeta_T(\rho)$ is strictly smaller than $\delta$ for any initial state $\rho \in \mathfrak D_{\ge \delta}$.
\end{proof}
\noindent We can now state the last two Lemmas. 
\begin{lem}
    There exists a $\gamma \in (0, \delta)$ such that, for some $T > 0$, $\zeta_T(\rho) < \delta - \gamma$ for all $\rho \in \mathfrak D_{> \delta - \gamma}$. \label{lem45mod}
\end{lem}
\noindent The proof follows the one of Lemma (4.5) in \cite{mirrahimi}, substituting the set $\mathfrak D_{> 1 - \xi}$ with the set $\mathfrak D_{> \delta - \xi}$ and exchanging the order of $V$ and $\mathbb E$ by linearity. The $T > 0$ of the new statement is the one obtained by applying Lemma (\ref{Texists}). In a similar fashion, we can obtain a revised version of Lemma (4.6) of \cite{mirrahimi}: \begin{lem}

    Using the $\gamma$ of Lemma (\ref{lem45mod}), let $\tau_{\rho}(\mathfrak D_{> \delta - \gamma})$ be the first exit time of the actual trajectory $\rho_t$ from $\mathfrak D_{> \delta - \gamma}$ starting from $\rho_0 = \rho$. Then 
    \begin{equation}
        \sup_{\rho \in \mathfrak D_{> \delta - \gamma}} \mathbb E[\tau_{\rho}(\mathfrak D_{> \delta - \gamma})] < +\infty
    \end{equation}
    This implies in turn that $\tau_\rho(\mathfrak D_{> \delta - \gamma})$ is a.s. finite.
    \label{lem46mod}
\end{lem}

\bibliography{ref}

\vfill

\end{document}

%% file: header_traditional.tex
\usepackage{amsfonts,amsmath,amssymb}
\usepackage[english]{babel}
\usepackage[normalem]{ulem}
\usepackage{hyperref}

\usepackage{amsthm}
\usepackage{graphicx}
\usepackage{subcaption,caption}
\usepackage{epsfig}
\usepackage{ulem}
\usepackage{tikz}
\usepackage{mathrsfs}
\usepackage{bbold}             
\usepackage{bm}
\usetikzlibrary{shapes,arrows}
\usetikzlibrary{matrix}
\usetikzlibrary{calc,patterns,angles,quotes,arrows,automata}
\graphicspath{ {./imgs/}}

\usepackage[linesnumbered,noend]{algorithm2e}
\RestyleAlgo{ruled}
\SetKwInOut{Input}{Input}
\SetKwInOut{Output}{Output}
\SetKwInOut{Parameters}{Parameters}
\SetKwIF{If}{ElseIf}{Else}{if}{:}{else if}{else}{endif}

\theoremstyle{definition}

\usepackage{mathabx}

\renewcommand{\bar}{\overline}

\theoremstyle{theorem}

\newtheorem{theorem}{Theorem}

\theoremstyle{definition}

\theoremstyle{remark}

\usepackage[numbers,sort&compress]{natbib}